\documentclass[letterpaper, 10 pt, conference]{ieeeconf}  
\IEEEoverridecommandlockouts 
\usepackage{amsmath,amsfonts}

\let\labelindent\relax
\usepackage{amssymb, amsthm}
\newtheorem{thm}{Theorem}
\newtheorem{cor}{Corollary}
\newtheorem{defn}{Definition}

\usepackage{mydefs,mymath,kernMach,diffge}
\newcommand{\domain}{\Real^{n_d}}
\newcommand{\subD}{\mathcal{D}}

\newcommand{\HspaceN}[1]{\mathcal{H}^{#1}}
\newcommand{\safeS}{\mathcal{S}}

\renewcommand{\kernMap}{{\vec{k}}}
\newcommand{\RBF}{\varphi}

\newcommand{\slkmFunc}{f_{\text{KM}}}

\newcommand{\sampSetSafe}{\sampSet^{\text{s}}}
\newcommand{\sampSetUnsafe}{\sampSet^{\text{u}}}
\newcommand{\setSafe}{\mathcal{S}^*}
\newcommand{\setUnsafe}{\mathcal{U}^*}
\newcommand{\tsafe}{\textit{safe}}
\newcommand{\tunsafe}{\textit{unsafe}}

\newcommand{\coverSet}{{\rm{C}}}
\newcommand{\coverSetClosed}{{\bar {\rm{C}}}}

\newcommand{\numPoly}{n_p}
\newcommand{\numKernPoly}{N_p}      
\newcommand{\numKernRBF}{N_b}       

\newcommand{\SKMBFlong}{shallow kernel machine-zeroing barrier function}
\newcommand{\SKMBF}{SKM-ZBF}
\newcommand{\ZBFlong}{zeroing barrier function}
\newcommand{\ZBF}{ZBF}

\usepackage{calc}
\usepackage{cases}
\usepackage{url}

\usepackage{float,color,graphicx}
\usepackage{tikz}
\usepackage{pgfplots}
\pgfplotsset{compat=1.15}

\graphicspath{{nnGeom/}}

\title{\LARGE \bf Geometry of Radial Basis Neural Networks for \\ 
Safety Biased Approximation of Unsafe Regions}
\author{Ahmad Abuaish$^{\dagger}$, Mohit Srinivasan$^{\ddagger}$,
  Patricio A. Vela$^{\dagger}$%
\thanks{This work was supported in part by the National Science
Foundation under Award S\&AS \#1849333, by DARPA PAI, and by KACST Fellowship.}%
\thanks{$^{\dagger}$ School of Electrical and Computer Engineering,
        Georgia Institute of Technology, Atlanta, USA
        {\tt\small aabuaish@gatech.edu, pvela@gatech.edu}}%
\thanks{$^{\ddagger}$ Ford Motor Company, Dearborn, USA
        {\tt\small mohit.s@ieee.org}}%
}

\begin{document}

\maketitle
\thispagestyle{empty}
\pagestyle{empty}

\begin{abstract}
Barrier function-based inequality constraints are a means to enforce safety specifications for control systems. When used in conjunction with a convex optimization program, they provide a computationally efficient method to enforce safety for the general class of control-affine systems. 
One of the main assumptions when taking this approach is the \textit{a priori} knowledge of the barrier function itself, i.e., knowledge of the safe set. In the context of navigation through unknown environments where the
locally safe set evolves with time, such knowledge does not exist. 
This manuscript focuses on the synthesis of a zeroing barrier function characterizing the safe set based on safe and unsafe sample measurements, e.g., from perception data in navigation applications. 
Prior work formulated a supervised machine learning algorithm whose solution guaranteed the construction of a zeroing barrier function with specific level-set properties. However, it did not explore the geometry of the
neural network design used for the synthesis process. This manuscript
describes the specific geometry of the neural network used for zeroing barrier function synthesis, and shows how the network provides the necessary representation for splitting the state space into \textit{safe} and \textit{unsafe} regions.
\end{abstract}

\section{Introduction \label{sec:intro}}
Invariant set-based safe control synthesis~\cite{blanchini1999set} has become a favorable technique to enforce safety, as it provides theoretical guarantees when used to augment base controllers \cite{Aa[2019]bf}. For closed dynamical systems, the invariant set is represented by the zero and positive level-sets of a continuously differentiable implicit function, typically termed \emph{zeroing barrier function} (ZBF). Subsequently, for controlled systems, \emph{control barrier functions} (CBFs) are introduced
to represent the invariant set.
In safety-critical applications, CBF-based controllers are designed to render safe regions in the state space a positively invariant set.
Prevailing use of CBFs is in a point-wise optimization program solved via quadratic programming (QP) \cite{Aa[2017]bf}. CBF-based QPs have been used in a wide range of applications in robotics such as collision avoidance~\cite{wang_CBFs_multirobot, wang_quadcopters}, multi-robot coordination and task satisfaction~\cite{multirobot_TRO, CBFs_LTL_TRO, lindemann_coupled}, and automotive applications~\cite{Aa[2017]bf}. 

Usually, a CBF is handcrafted based on \textit{a priori} knowledge of the safe regions in the state space. However, there are applications where the safe regions evolve with time, with navigation being one. In these applications, it is critical to synthesize the CBF online using sensor measurements. Unfortunately, determining the true invariant region defined by a CBF is generally an NP-hard problem, but some techniques exist that can estimate the invariant region \cite{YuAa[2021]CBF}. Since the existence of a zeroing barrier function is needed to formally confirm the existence of a CBF,  this manuscript solely focuses on zeroing barrier function synthesis to separate state space into \emph{safe} and \emph{unsafe} regions.
The {\ZBF} is constructed from a two-layer kernel machine network trained from a labeled dataset of {\tsafe} and {\tunsafe} samples. 
Further, the geometry of the kernel functions in the network is explored to efficiently partition the space. The approach is motivated by the Kolmogorov–Arnold representation theorem, which implies that two-layer neural networks may be capable of approximating continuous functions \cite{Jo[2021]kathm}.

Recently, several data-driven approaches for constructing a CBF were proposed to account for uncertainties in either the system dynamics, unsafe regions, or both.  
One approach category for learning CBF's involves supervised offline learning. Instances include 
imitation learning where training data is generated by an expert actor or optimal control simulations \cite{Al[2020]cbf-oc, YuCh[2021]cbf-imitation}. The offline nature lacks the ability to accommodate real-time changes in the environment. In contrast, \emph{self-supervised} approaches permit online learning. 
Initial work on self-supervised Bayesian learning system of uncertain dynamics \cite{ViNi[2021]cbf-dyn} with known barrier functions was  merged with \cite{KeNi[2022]cbf_sdf} to learn the system dynamics and an implicit function representation of the unsafe region \cite{KeNi[2022]cbf-dyn-sdf}.  In \cite{KeNi[2022]cbf_sdf}, a signed distance function representing obstacles is modeled as a deep neural network trained from range sensor data via stochastic gradient descent (SGD) with replay memory. Similarly, \cite{FaLi[2016]cbf} presented the construction of a probabilistic occupancy map from a kernel-based logistic regression model trained from range sensor data via SGD. However, there were no hard constraints on misclassifying unsafe data points in the underlying SGD optimization process, which nullifies any possible theoretical safety guarantee.

Our previous work focused on creating a {\ZBF} for navigation applications based on data collected from LiDAR sensors \cite{MoVe[2020]bf}.
A two-layer network with Gaussian radial basis functions (GRBFs) was synthesized from this data. The first layer used sparsely distributed (over the domain) GRBF centers, while the second GRBF layer was learnt during the kernel support vector machine (kSVM) optimization process. The kSVM optimization specifications provided formal guarantees regarding the partitioning of the domain into {\tsafe} and {\tunsafe} regions. However, the work did not discuss the geometry and associated properties of the GRBF network. This work analyzes the geometry of the two-layer network and the structure of the optimization problem to prove the existence of a {\ZBF} with known partitioning properties.



The manuscript organization is as follows: Section II discusses the geometry of Gaussian kernel functions with respect to the first layer. 
Section III covers the second layer construction, geometry, and optimization specifications. 
Section IV discusses the qualification of the two-layer kernel machine network to be a {\ZBFlong} along with a kernel basis selection strategy. 
Sections V and VI present case studies and concluding remarks, respectively.

\section{Geometry of the Kernel Hilbert Space\label{sec:mapGeom}}
GRBF Neural Networks (GRBF-NNs)  have several properties ideal for {\ZBFlong} creation.
First they are universal approximators \cite{JW[1991]GRBF},
second they exhibit locality \cite{Mc[2007]GRBF}, and third they
partition space.  
Since the last property is less frequently mentioned, but often used, this section devotes attention to space partitioning, as it is essential for data-driven safe set generation. 

The focus of this section is on GRBF-NNs as kernel machines whose kernel functions are radial basis functions (RBFs), which generally take the following form,
\begin{equation}
  \kernFunc(x_i, x_j) = \RBF(\norm{x_i - x_j}; \sigma),\quad\text{for}\quad
    \function{\RBF}{\Real^+}{\Real^+},
\end{equation}
where $x_i,x_j\in\Real^{n_d}$, $\norm{\cdot}$ is a norm and $\sigma$ is the scalar bandwidth, which influences the
sensitivity of the basis function $\RBF$. A kernel machine based on radial
basis functions generates a mapping through the use of multiple kernel
function mappings with fixed argument elements $\centc_j$ in a
center set $\centerSet$. For $\centc_j \in \centerSet$ the
kernel mapping is
\begin{equation} \eqlabel{kernMap}
  \function{\kernMap}{\domain}{\HspaceN{\centerNum}} ,\ \ 
    x \mapsto [\kernFunc(x,\centc_1)\; \cdots \;\kernFunc(x,\centc_{\centerNum})]^T.
\end{equation}
where $\centerNum = \card{\centerSet}$ and $\HspaceN{\cdot}$ indicates
that the output space is a Hilbert space. For this kernel machine to
define a scalar function requires specifying 
$\alpha \in \Real^{\centerNum}$ such that
\begin{equation} \eqlabel{kmFunc}
  \slkmFunc(x) = \innprod{\alpha}{\kernMap(x)} = \alpha^T \kernMap(x).
\end{equation}
Kernel machines permit more general function classes than RBFs in $\Real^{n_d}$ (subsequent sections will use polynomial
kernel functions). That said, RBFs have geometric properties that are 
implicitly exploited in kernel machine learning applications.
Specializing to the case of the Gaussian radial basis function (or Gaussian kernel), let the kernel function be 
\begin{equation}
    k_G(x,c)=\exp\left(-\frac{\norm{x-c}^2}{\sigma^2}\right).
\end{equation}
\begin{thm}\label{th:ker-map}
The kernel mapping, with $\centerNum$ Gaussian kernels, maps the
input domain $\subD \subset \Real^{n_d}$ into a surface in the Hilbert
space $\HspaceN{\centerNum} \subset \Real^{\centerNum}$ when
$\centerNum \ge n_d + 1$ and there are $n_d + 1$ centers capable of defining a
coordinate system in $n_d$.
\end{thm}
\begin{proof}
Showing that the kernel mapping is 1-1 establishes this property.
The pre-image, $\kernFunc^{-1}(\cdot, c_i)$, of each coordinate's kernel
mapping is a sphere in the input space. The intersection of all spheres
for all coordinate mappings establishes the pre-image point, which is
unique only if the intersection is unique. 
Finite solutions for the intersections has been proven in the context of
rigid body geometry for $\centerNum = n_d$  \cite{LyPa[2017]ModRobo},
with $\centerNum \ge n_d + 1$ necessary for a single valid solution.  The
requirement for the centers is that they lead to a basis of $n_d$
vectors when using one of the points as the origin and using 
$n_d$ other points to obtain the basis vectors relative to that
origin.  Each pre-image imposes a constraint on the degrees of freedom
of the input point $x \in \domain$, such that the $n_d+1$ intersecting
pre-image spheres do so at a single point. 
When $\centerNum > n_d+1$, the additional pre-image constraints are
redundant and effectively impose no constraints. The rank of the mapping is 
$n_d < \centerNum$, hence it maps to a surface of dimension $n_d$.
\end{proof}

Theorem \ref{th:ker-map} applies to any RBF with infinite support; an RBF with
finite support will have similar properties but will include an $\epsilon$-covering constraint.
Basis functions generated from other norms also have similar properties
but may require more centers to induce a 1-1 mapping.
If the kernel is differentiable, then the 1-1 mapping is an embedding. Safe set generation with an appropriate kernel mapping
will involve defining the concept of a {\em kernel
embedding}\footnote{Not to be confused with the \textit{kernel mean
embedding} which is for probability distributions \cite{Kr[2017]kerMean}.}  
and its associated {\em kernel embedding inducing} data set.

\begin{defn}
The kernel mapping
$\function{\kernMap}{\domain}{\HspaceN{\centerNum}}$ 
is a {\em kernel embedding} if the center set
$\centerSet \subset \subD$ 
generating the kernel mapping are such that it is 1-1 and the kernel
function $\function{\kernFunc}{\domain \times \domain}{\Real}$ is
differentiable. The set of centers is called the 
{\em kernel embedding inducing set} (KEI set).
\end{defn}

\begin{cor}
Consider a finite set of points $\sampSet_p \subset \Real^2$ with an
associated triangulation.  Under a kernel embedding, each triangular
region maps to a surface in $\HspaceN{\centerNum}$ homeomorphic to a 2-simplex.
\end{cor}

\begin{cor}
Consider a finite set of points $\sampSet_p \subset \Real^3$ with an
associated tetrahedralization.  Under a kernel embedding, each
tetrahedron maps to a surface in $\HspaceN{\centerNum}$ homeomorphic to
a 3-simplex.
\end{cor}

If the domain $\subD \subset \domain$ is a set of disconnected regions
excluding points at infinity, then a kernel embedding will map to a set
of disconnected surfaces.  Similarly, a collection of non-intersecting
triangulations/tetrahedralizations maps to a set of disconnected
surfaces under a kernel embedding.

\subsubsection{Geometry of a Kernel Embedding}
The Gaussian kernel mapping outputs lie in the unit cube of 
$\HspaceN{\centerNum} \subset \Real^{\centerNum}$.
Each center $c_i$ in the set $\centerSet$ maximizes its associated
coordinate (evaluates to 1), which means that there is a neighborhood of
$c_i$ in $\subD$ for which this same coordinate is also maximal for all
points in the neighborhood.  Points tending to infinity map to the
origin in $\HspaceN{\centerNum}$ since the Gaussian radial basis
function tends to zero as the input radius tends to infinity.

\begin{cor} \corlabel{compactSurface}
For a kernel embedding defined using the Gaussian kernel, a compact
domain $\subD$ maps to a compact surface whose points lie outside of a
ball centered at the origin.
Furthermore, $\kernMap(\Real^{n_d}) \cup \{\vec 0\}$ is a compact surface.
\end{cor}

With the surface geometry of the kernel embedding $\kernMap$ for a
given KEI set $\centerSet$ established, the next step
is to consider partitions of the surface in $\HspaceN{\centerNum}$.
Given that the
unbounded input space maps to a compact surface, defining a partition of
the space is equivalent to defining a cutting surface transverse to the
compact surface that divides space into positive and negative regions.  
The intersection of the compact surface with the cutting surface in
the Hilbert space maps back to separating boundaries in the original space.
The objective will be to define cutting surfaces by specific level-sets
of an implicit function in the Hilbert space.  In short, safe region
generation in the input space involves construction of an implicit
function in the Hilbert space based on sampled points with 
{\tsafe}\,/\,{\tunsafe} labels.

\begin{figure}[t]
  \newlength{\xsh}
  \newlength{\ysh}
  \newlength{\bht}
  \newlength{\bhb}
  \newlength{\ylabs}
  \setlength{\bht}{1.17in}
  \setlength{\bhb}{1.00in}
  \setlength{\xsh}{10pt}
  \setlength{\ylabs}{-6pt}
  \vspace*{5pt}
  \begin{tikzpicture}[inner sep=0pt, outer sep=0pt]
    \node (fig_a1) at (0in,0in)
      {\includegraphics[height=\bht]{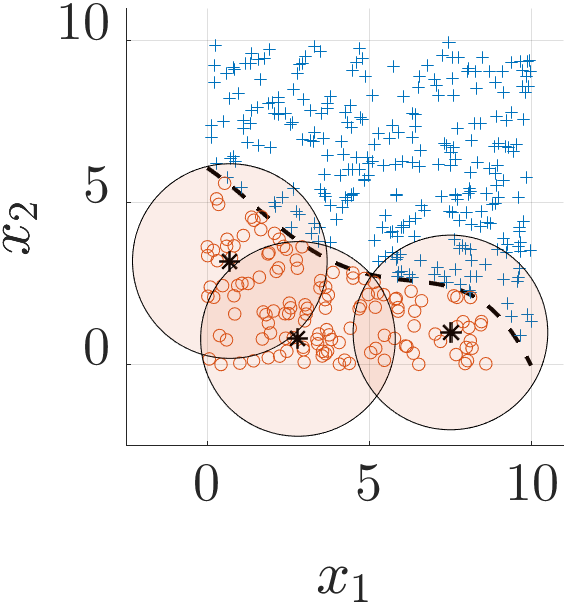}};
    \node[anchor=west] (fig_b1) at ([xshift=\xsh]fig_a1.east)
      {\includegraphics[height=\bht,clip=true,trim=1.7in 0in 0in 0in]{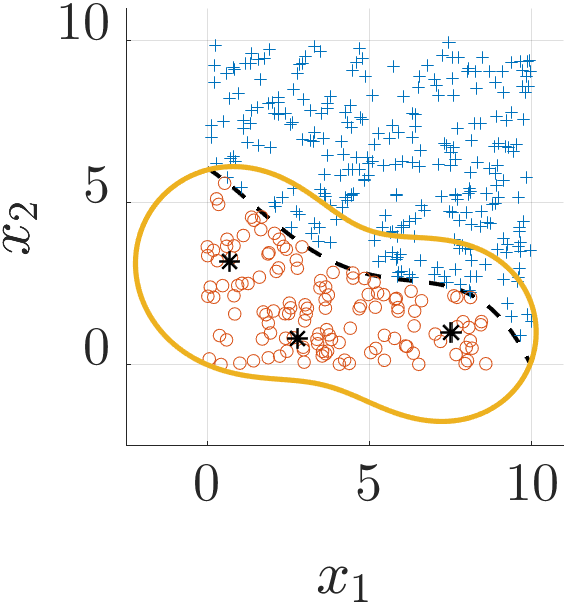}};
    \node[anchor=west] (fig_c1) at ([xshift=\xsh]fig_b1.east)
      {\includegraphics[height=\bht,clip=true,trim=1.7in 0in 0in 0in]{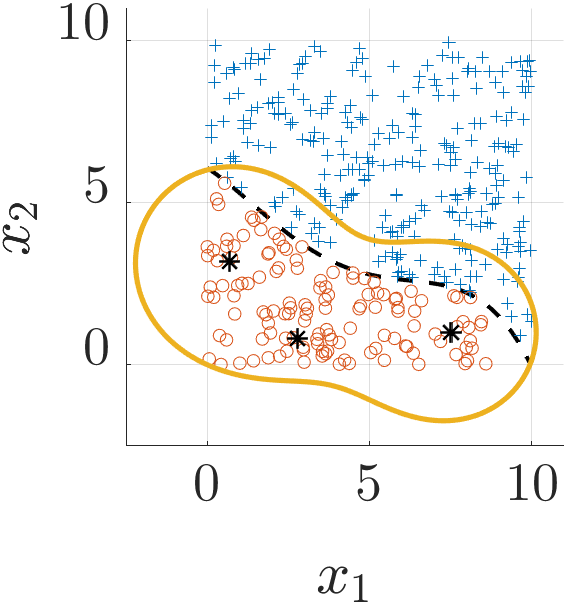}};
      
    \node[anchor=north west] (fig_a2) at ([yshift=\ysh]fig_a1.south west)
      {\includegraphics[height=\bhb,clip=true,trim=0.00in 0in 0in 0in]{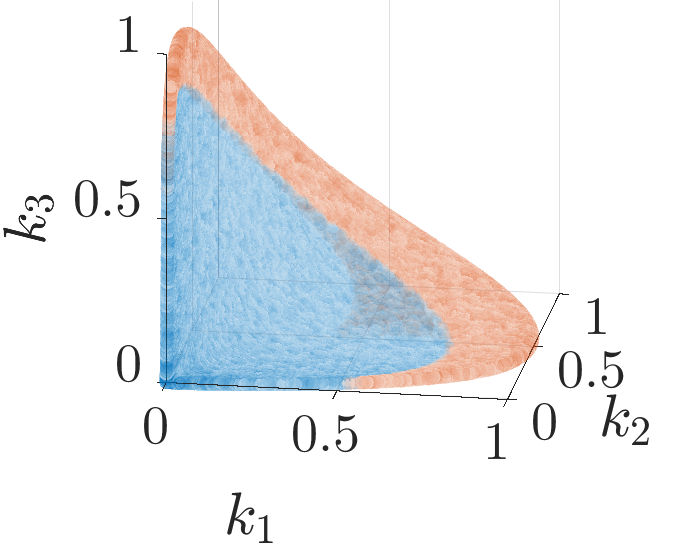}};
    \node[anchor=north west,xshift=-5pt] (fig_b2) at (fig_b1.south west)
      {\includegraphics[height=\bhb,clip=true,trim=0.9in 0in 0in 0in]{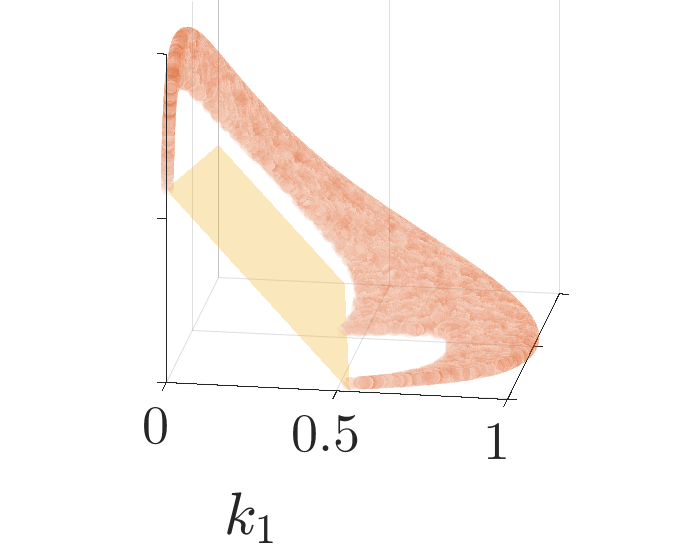}};
    \node[anchor=north west,xshift=-4pt] (fig_c2) at (fig_c1.south west)
      {\includegraphics[height=\bhb,clip=true,trim=0.9in 0in 0in 0in]{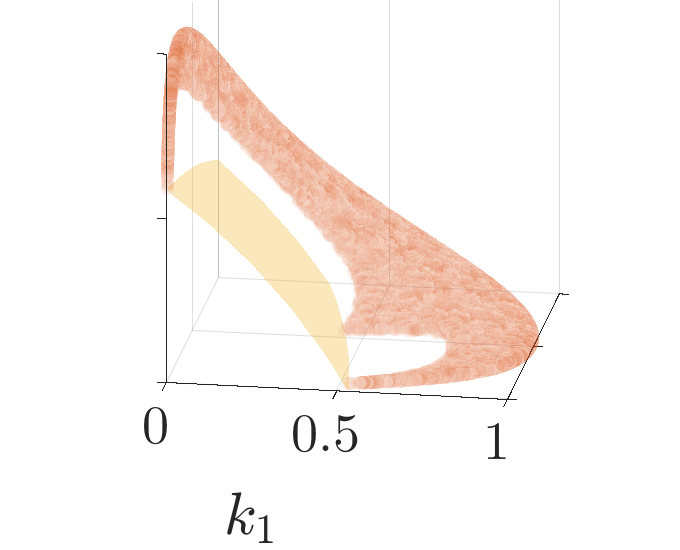}};


    \node[anchor=north west,yshift=-3pt] (legend) at (fig_a2.south west)
      {\includegraphics[width=\linewidth]{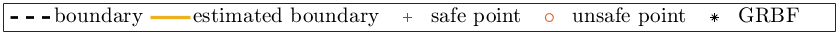}};
    
  \end{tikzpicture}
 
 \caption{Depiction of the input domain and the associated Hilbert space along with the GRBF centers and the separating boundaries and surfaces.
 \figlabel{example}}
\end{figure}

The observation leading to Corollary \corref{compactSurface} hints that
the geometry induced by the KEI set plays a role in establishing the
surface image in the Hilbert space.  Here, we connect the KEI set to
coverings of the input domain.  For the sets and spaces defined here,
define an $\epsilon$-covering as follows:

\begin{definition} For $\domain$ with the Euclidean norm (2-norm), let
$\subD \subset \domain$. For $\epsilon > 0$, the set 
$\centerSet \subset \subD$ is an $\epsilon$-cover for $\subD$ if for
every $x \in \subD$ there exists a $c_i \in \centerSet$ such that
$\norm{x - c_i} \le \epsilon$. Equivalently,
$\subD \subset \union_{i} \mathcal{B}_\epsilon(c_i)$.
\end{definition}

In the context of a Gaussian radial basis neural network or Gaussian
kernel machines, $\epsilon$-covers partition the original space in a
predictable manner in the kernel mapped Hilbert space.  An equivalent 
$\epsilon$-cover specification is the following $\infty$-norm inequality
applied to the kernel mapping output:
\begin{equation} \eqlabel{epsCovKM}
  \coverSet = 
  \setof{x \in \Real^{n_d}}{\norm{\kernMap(x)}_\infty 
                                              \ge \RBF(\epsilon; \sigma) = e^{-\epsilon^2/\sigma^2}}.
\end{equation}
If the KEI set is defined based on an $\epsilon$-cover of $\subD \subset \domain$, then the kernel map will map points
in $\subD$ to points in $\HspaceN{\centerNum}$ on a surface some minimal
distance from the origin.  Points in $\domain$ receding from $\subD$
will tend towards the origin in $\HspaceN{\centerNum}$. 

The top row in Figure \figref{example}
depicts a 2D example consisting of a collection of safe (blue +) and unsafe
(orange $\circ$) points generated from a nonlinear separating boundary
(dashed black curve), as well as three GRBFs (black *) centered in the unsafe region, $\sampSet^u \subset \subD$, that are an $\epsilon$-cover.
These three centers create a mapping of the 2D domain into a 3D
Hilbert space, depicted in the bottom row where 
the orange surface is the unsafe region surface $\kernMap(\setUnsafe)$ and 
the blue surface is the safe region surface $\kernMap(\setSafe)$. 
The cutting surface sought is one that separates points ``near'' to the
origin from those in $\kernMap(\setUnsafe)$.  Two such surfaces are
depicted in the second and third columns (bottom row) along with their level-set boundaries in the input space (top row).  The next section describes how to derive cutting surfaces from solutions to constrained optimization problems.


\section{Cutting Surfaces and Partitions}

This section exploits the geometry of the kernel mapping to define
cutting hyperplanes in its output Hilbert space. These hyperplanes translate 
to partitions of the input space. Initially, the cutting hyperplanes will
be in the original Hilbert space and are defined by a single-layer kernel
machine network. Next, a second layer is added to the kernel machine network to create cutting hyperplanes that translate to elliptical volumes in the original Hilbert space. More complex cutting surfaces are then explored
based on the inclusion of positive (\textit{unsafe}) and negative
(\textit{safe}) samples. 

\subsection{Single-Layer Kernel Machine Network Partitions}

Constructing an implicit partition function with barrier function
properties with a single-layer kernel machine involve identifying a
cutting hyperplane for the collection of safe and unsafe data. 
Consider a kernel embedding built using the {\tunsafe} data where
$\centerSet$ is a KEI set and covering of the unsafe set for some $\epsilon$.  
One cutting hyperplane results from the following linear program:
\begin{equation}  \eqlabel{linProgSL}
\begin{aligned}
    \min_{\alpha\in\Real^{\centerNum}} \quad &\vec 1\,^T \alpha\\
    \text{s.t.} \quad & \alpha^T \vec k(x_i) \geq 1, \; \forall i \in\setof{l}{x_l\in\sampSet^u}\\
                \quad & \alpha_j \geq 0, \; \forall j=1,\hdots,n_c
\end{aligned}
\end{equation}
where $\vec 1=[1\, \cdots\, 1]^T\in\Real^{\centerNum}$, and
$\alpha\in\Real^{\centerNum}$ are the coefficients of the separating hyperplane.
The one level-set of $\slkmFunc(x) =\alpha^T\vec k(x)$ defines the
hyperplane in the positive hyperoctant that is furthest from the origin
and places all unsafe points in $\sampSet$ in the positive half-plane.
The linear program is guaranteed to have a solution:

\begin{thm}
Given a kernel embedding $\kernMap$ defined from an $\epsilon$-covering
of $\sampSetUnsafe$, there is a hyperplanar splitting of
$\HspaceN{\centerNum}$
described by
\begin{equation}
  \safeS = \setof{x \in \subD}{ \innprod{\vec b}{\kernVec{}(x)} < 1 } 
\end{equation}
where $\safeS \subset \setSafe \subset \subD$ and $\sampSetUnsafe \subset
\bar \safeS$.
\end{thm}
\begin{proof}
Consider the covering set of centers $\centerSet$. 
Generate a clustering of the unsafe data points $\sampSetUnsafe$ based 
on the cluster assignment function $x^u \in \sampSetUnsafe
\mapsto \argmax_i k(x^u, c_i)$ for $c_i \in \centerSet$. For each
cluster set $\sampSetUnsafe_i$ find the minimum value $y^i=k(x,c_i)$ for $x \in \sampSetUnsafe_i$. Each cluster
minimum defines the $i$-th coordinate intercept for a cutting plane in
the Hilbert space, which defines $\vec b$. All \textit{unsafe} points
$\sampSetUnsafe$ lie on the non-negative side, possibly with some
\textit{safe} points from $\sampSetSafe$. Only \textit{safe} points in
$\sampSetSafe$ lie on the negative side (it is half-plane containing origin).
%
\end{proof}

\begin{cor}
By virtue of defining a splitting with a kernel embedding, this same
operation generates a splitting of the original domain $\subD$, and
likewise the full space $\domain$.
\end{cor}

The second column of Figure \figref{example} depicts the separating
hyperplane, generated by the linear program in \eqref{linProgSL}. The bottom plot shows the hyperplane (yellow surface) in the Hilbert space, and the top plot shows the separating boundary (yellow closed contour) in the input domain.

\subsection{Two-Layer Neural Network Partitions}

The Hilbert space splitting based on a separating hyperplane may not
capture the true classification boundary in the Hilbert space. Consequently, a richer decision boundary based on a nonlinear separating surface
should create a splitting with fewer misclassified points.
Consider the problem of defining a quadratic safety volume in the Hilber
space specified by a quadratic boundary
\begin{equation} \eqlabel{quadSurfEqn}
  \zeta^T A_2 \zeta + b_2^T \zeta + c_2 = 0,
\end{equation}
where elements below the surface (evaluate to less than 0) are
guaranteed to be {\em safe}, and {\em unsafe} elements are guaranteed to
be above the surface (evaluate to greater than 0). 
Now, the optimization problem becomes a quadratically constrained linear program (QCLP), which is an NP-hard problem in theory \cite{Je[2005]qcqp,kr[2020]qcqp}. Although relaxation techniques for solving QCLPs exist, such optimization problems should be avoided when possible.
Adding a quadratic polynomial kernel as a second layer to the network will preserve the LP formulation in \eqref{linProgSL} for the cutting surface optimization problem.


\subsubsection{A Quadratic Polynomial Layer}
Using a quadratic polynomial kernel of the form 
$p_2(x,y) = (x^T y + \lambda)^2$ leads to quadratic cutting surfaces
in $\HspaceN{n_q}$ where $n_q$ is the number of kernels. 
The mapping is now ${\vec{p}}_2 \circ \vec k : \subD \rightarrow \HspaceN{n_q}$, where 
\begin{equation}\eqlabel{poly2KerVec}
    \vec{p}_2(x) = \left[p_2(x,y_1) \; \cdots \; p_2\left(x,y_{n_q}\right)\right]^T.
\end{equation}

For the quadratic polynomial layer, the parametric degrees of freedom
are the vectors $y_i$ such that each kernel coordinate is
$p_2(x,y_i; \lambda)$. For simplicity, let 
$n_q \ge \centerNum$ such that $y_i = e_i$ for $i \in \set{1, \dots, \centerNum}$
where $e_i$ is the $i^{\text{th}}$ unit coordinate vector. The remaining
vectors $y_i$, if any, for $\centerNum < i \le n_q$ can be any basis elements that support the task at hand (typically selected based on the training data). 
Creating a separating boundary is equivalent to establishing the coefficients $\alpha_i$ for the following constraint equation,
\begin{equation} \eqlabel{quadSurfTL}
  z^T \of{ \sum_i \alpha_i y_i y_i^T} z 
    + 2 \lambda \of{ \sum_i \alpha_i y_i^T} z 
    + \lambda^2 \of{ \sum \alpha_i} = 0,
\end{equation}
where $z \in \HspaceN{\centerNum}$. The structure of \eqref{quadSurfTL} matches that
of \eqref{quadSurfEqn}. Note that due to the limited number of kernels, and thus limited number of basis vectors $y_i$, \eqref{quadSurfTL} can only represent a subset of solutions generated by \eqref{quadSurfEqn}. However, the trade-off is that the optimization problem for solving the coefficients of kernel machines remains a linear program:
\begin{equation}  \eqlabel{linProgTL}
\begin{aligned}
    \min_{\alpha\in\Real^{n_q}} \quad &\vec 1\,^T \alpha\\
    \text{s.t.} \quad & \alpha^T \vec{p}_2 \circ \vec k(x_i) \geq 1, \; \forall i \in\setof{l}{x_l\in\sampSet^u}\\
                \quad & \alpha_j \geq 0, \; \forall j=1,\hdots,n_q
\end{aligned}
\end{equation}

\begin{thm} \thmlabel{existQL}
Given a kernel embedding $\kernMap$ defined from an $\epsilon$-covering
of $\sampSetUnsafe$, there is a hyperspherical splitting of
$\HspaceN{\centerNum}$ described by
\begin{equation}
  \safeS = \setof{x \in \subD}{\norm{\kernVec{}(x) + \lambda \vec 1}_2
    < \rho_P\ \text{for}\ x \in \sampSetUnsafe}
\end{equation}
where $\safeS \subset \setSafe \subset \subD$, 
$\lambda \ge 0$,
$\rho_P = \min\left( \norm{\kernVec{}(x) + \lambda \vec 1}_2 \right)$ for 
$x \in \sampSetUnsafe$, 
and $\sampSetUnsafe \in \bar \safeS$.
\end{thm}
\begin{proof}
The theorem asserts the existence of a feasible quadratic surface with
hard \textit{unsafe} constraints for splitting the Hilbert space using
the 2-norm.  The 2-norm operation is equivalent to setting $\lambda =
0$, $\alpha_i = \rho_P^{-2}$ for $i \le \centerNum$,
and $\alpha_i = 0$ for $\centerNum < i \le n_q$ (should such $i$ exist)
for the polynomial quadratic kernel problem specified by
\eqref{linProgTL}.  Thus, the set of feasible solutions to
\eqref{linProgTL} for $\lambda = 0$ has at least one element in it.
For $\lambda > 0$, $\rho_P = \min\left( \norm{\kernVec{}(x) + \lambda \vec 1}_2 \right)$ 
and gives a hyper-sphere centered at $-\lambda \vec 1$.  
\end{proof}

The theorem establishes the existence of a solution to the linear program
defined in \eqref{linProgTL}, thereby showing that the solution space is non-empty
when $\lambda \ge 0$.  The case of $\lambda < 0$ is possible but trickier to solve
for, and may lead to poor solutions when $-1 < \lambda < 0$. 
Referring again to the example in Figure \figref{example}, the
third column depicts the separating quadratic surface, solved by the linear
program in \eqref{linProgTL}. The bottom plot shows the quadratic surface (yellow surface) in the Hilbert space, and the top plot shows the separating boundary (yellow closed contour) in the input domain. There are slightly fewer misclassified safe data points compared to the linear hyperplane results in the second column.


Higher order polynomial kernels may be used to partition the space.
Doing so should provide more space carving degrees of freedom but will
require a policy for selecting the kernel basis elements (beyond the
unit coordinate vectors of the Hilbert space). 
Solving for the second layer as a kernel SVM using a sequential minimal optimization (SMO) algorithm
\cite{Platt[1998]smo} will provide a solution that identifies the basis
elements and solves for their coefficients, much like in \cite{MoVe[2020]bf}.
This section focused on a constructive approach to guarantee a non-empty solution space; it can be augmented with an SMO-like basis expansion solver to obtain $n_q >
\centerNum$.

\subsection{Two-Layer Partitions with Safe and Unsafe Samples}\subseclabel{safe-unsafe}

The linear programs for the splitting hyperplane and hyperellipsoid in the
lifted Hilbert space $\HspaceN{n_c}$ attempt to find the furthest hyperplane or
largest hyperellipsoid. They may be conservative relative to other
options for synthesizing such a boundary due to the curvature of the
Hilbert space or the shape of the separating boundary. 
As a result, misclassification of safe points as unsafe will occur, as seen in Figure \figref{example}.
Richer surfaces exist that better capture the regions misclassified by the linear and ellipsoidal cutting surfaces.
However, the linear program establishing the model coefficients $\alpha$ will require knowledge of what parts of the mapped space should
be considered \textit{safe} versus \textit{unsafe}, and require
having sufficient samples of both classes.  These samples permit
recovery of more complex cutting surfaces.


The coordinate vectors, $y_i$, for the polynomial kernels in the multi-order polynomial layer should at least include $n_c$ basis unit vectors for each polynomial order (i.e., the minimum LP problem size grows linearly with the
maximum polynomial order). These can be complemented with a specialized or 
task-specific basis expansion criteria. The polynomial kernel mapping is the following,
\begin{equation}\eqlabel{p_vec}
\begin{aligned}
  \vec{p}(z) = [&(y_1^T z + \lambda_1)\ \cdots\  (y_{n_c}^T z + \lambda_1)\ \cdots\ \\
                &(y_1^T z + \lambda_{p_i})^{p_i}\ \cdots\ (y_{n_c}^T z + \lambda_{p_i})^{p_i}\ \cdots]^T,
\end{aligned}
\end{equation}
where $p_i \in \set{1, \dots, \numPoly}$, $\numPoly$ is the maximal polynomial
order used, and $\numKernPoly$ is the total number of polynomial kernels
used across all orders. The minimum basis elements requirement means
that $\numKernPoly \ge \numPoly \centerNum$. The linear-in-$\alpha$ constraint 
equation for a multi-order polynomial surface cut is
\begin{equation}
\alpha^T p(z) = 
  \sum_i \alpha_i \of{y_i^T z + \lambda_{p_i}}^{p_i} = 0.
\end{equation}
Per \cite{MoVe[2020]bf}, the inclusion of positive and negative samples
should involve hard constraints on the \textit{unsafe} points and soft
constraints for the \textit{safe} points
\begin{align} \eqlabel{posnegLP}
  \min_{\alpha\in\Real^{N_p}, \xi\in\Real^{n_s}} \quad & \vec 1\,^T\xi\\
  \nonumber
  \text{s.t.} \quad 
      & \alpha^T \vec{p} \circ \vec k(x_i) \geq 1,\ \forall i\in\setof{l}{x_l\in\sampSet^u}
      \\ \nonumber
      & \alpha^T \vec{p} \circ \vec k(x_j) \leq -1+\xi_j,
      \\  \nonumber
      & \xi_j \geq 0\, \quad \qquad \qquad \forall j\in\setof{l}{x_l\in\sampSet^s}
      \\ \nonumber
\end{align}

\begin{thm} \thmlabel{TLpolySafe}
Defining the two layer neural network $f = \vec{p} \circ \vec k$ with
$\numPoly \ge 2$ and performing a hyper-planar partitioning of the space in
the output space provides an equivalent or better splitting of the
space, based on {\tsafe} misclassification counts. 
\end{thm}
\begin{proof}
The solution space for \eqref{posnegLP} contains the linear and
quadratic cases by setting the appropriate coefficients $\beta_i$ to
zero. The optimization problem will either return one of these options
or it will find a better one. By virtue of having a slack minimizing
cost, better solutions will involve either the same slack or smaller.
If more slack variables evaluate to zero, there will be
less \textit{safe} point classification errors.
\end{proof}

The original two-layer formulation in \cite{MoVe[2020]bf} uses a 2-layer
GRBF network created using a kernel-SVM process for the second layer.
The Gaussian kernel in the second layer defines an alternative nonlinear
space for generating the cutting surface, which may provide a richer cutting surface solution space.  The theorems and corollaries
here provide a more formal analysis of the partitioning properties of
the 2-layer kernel machine, which is an instance of a shallow layer
kernel machine (SKM), and the specification of learning problems on the
output layer that are linear program formulations with a non-empty
solution space. Similar results apply for a Gaussian kernel second layer.

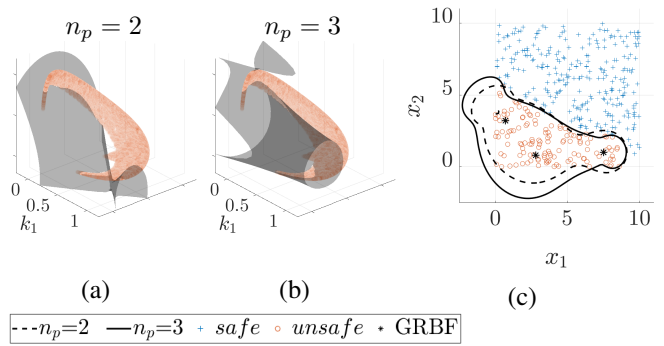
\begin{figure}[t]
  \setlength{\bht}{1in}
  \setlength{\xsh}{3pt}
  \setlength{\ysh}{-3pt}
  \setlength{\ylabs}{-6pt}
  \vspace*{5pt}
  \begin{tikzpicture}[inner sep=0pt, outer sep=0pt]
    \node(fig_a) at (0in,0in)
      {\includegraphics[width=\bht]{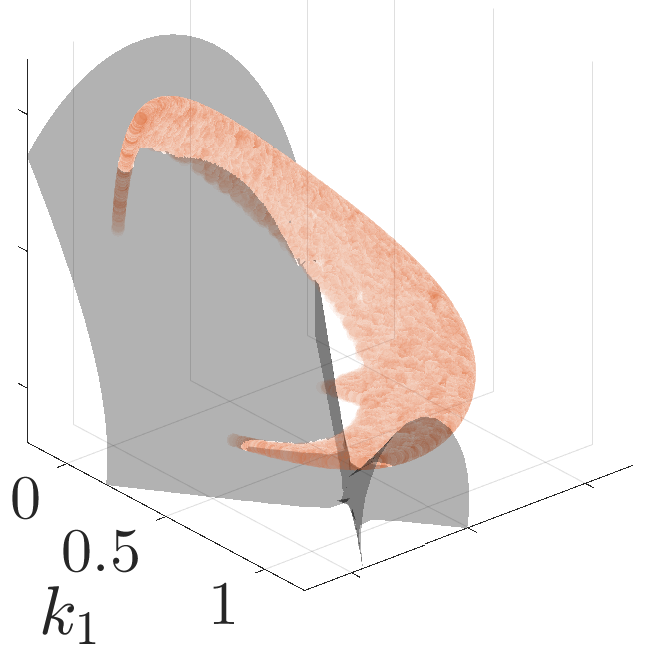}};
    \node[anchor=west] (fig_b) at ([xshift=\xsh]fig_a.east)
      {\includegraphics[width=\bht]{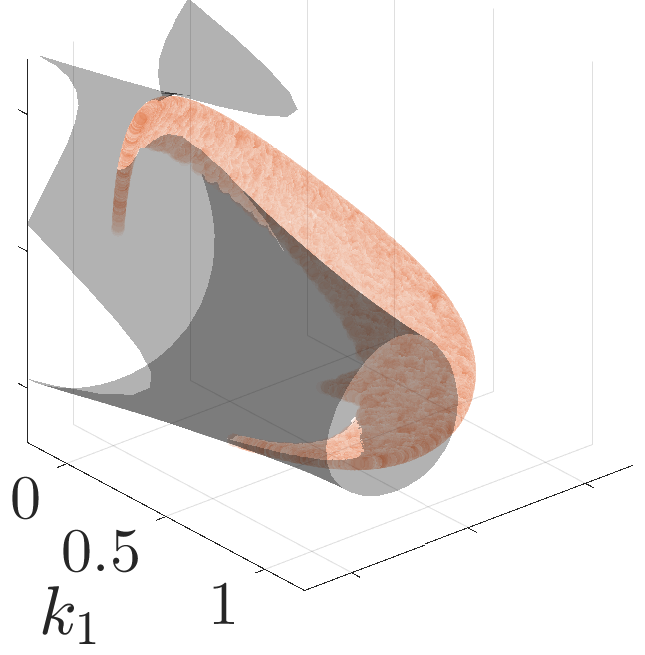}};
    \node[anchor=west] (fig_c) at ([xshift=\xsh]fig_b.east)
      {\includegraphics[width=1.3in,clip=true]{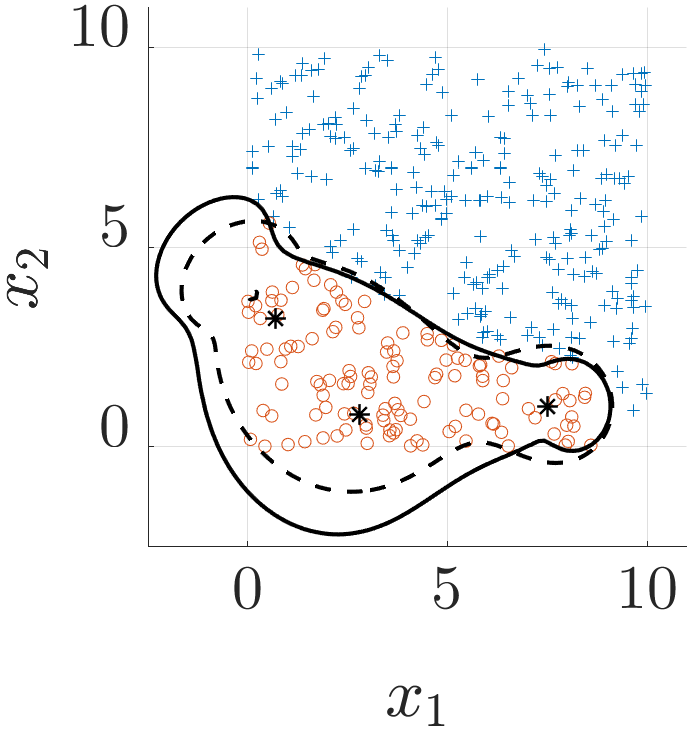}};
      
    \node[anchor=north,xshift=-3pt,yshift=-17] at (fig_a.south)
      {(a)};
    \node[anchor=north,xshift=-3pt,yshift=-17]  at (fig_b.south)
      {(b)};
    \node[anchor=north,xshift=-3pt,yshift=\ylabs] at (fig_c.south)
      {(c)};

    \node[anchor=south] at (fig_a.north){$n_p=2$};
    \node[anchor=south] at (fig_b.north){$n_p=3$};

    \node[anchor=west,yshift=-1in] at (fig_a.west)
      {\includegraphics[width=0.7\linewidth]{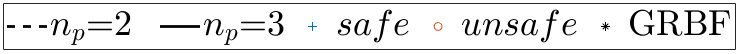}};
    
  \end{tikzpicture}
  
 \caption{Depiction of cutting surfaces and separating curves for two-layer, multi-order polynomial partitions using safe and unsafe data points. 
   \figlabel{exampleTL}}
\end{figure}

\paragraph{Discussion}
The provision of negative ({\tsafe}) samples changes the structure of the problem by adding a point set complementary to the unsafe set, deemed to be safe.
Incorporating this data into the optimization problem adds
information about what regions of the Hilbert space $\HspaceN{n_c}$
should lie on the negative side of the boundary,
thereby pushing the boundary outwards when there are negative samples on
the other side of the hyperplanar or hyperellipsoidal splitting (based on
available degrees of freedom in the polynomial kernel layer).
Solutions should improve the accuracy of the estimated bound and reduce the quantity of {\tsafe} misclassification errors, per Theorem \thmref{TLpolySafe}.

Figure \figref{exampleTL} depicts the cutting surfaces for the case of $\numPoly=2$
(left plot) and $\numPoly=3$ (middle plot), and their input domain boundaries (right plot) 
for the same example problem data of Figure \figref{example}. Both surfaces were generated using positive and negative samples per \eqref{posnegLP}. There are less misclassified \textit{safe} data points, with the $\numPoly=3$ case achieving zero slack cost.

\paragraph{Maximally Flexible Shallow Network Design}
The constructions to date are  based on a covering of the positive set 
(\textit{unsafe} set).  For potentially more accurate results, additional 
basis centers can be chosen from the negative (\textit{safe}) set. 
Extremizing coordinates associated to these centers indicate movement
towards \textit{safe} regions and away from \textit{unsafe}. 
Their inclusion changes the structure of the optimization since these are
points in the Hilbert space surface to avoid. The optimization problem 
specification should seek to generate a cutting surface that carves out such regions. 
Because of the flipped nature, a bias term will be needed. The bias term can be easily appended to the multi-order polynomial layer mapping, updating $\vec{p}$ from \eqref{p_vec} to 
\begin{equation}\eqlabel{p_vec_bias}
\begin{split}
  \vec{p}(z) = [&(y_1^T z + \lambda_1)\ \cdots\  (y_{n_c}^T z + \lambda_1)\ \cdots\ \\
                &(y_1^T z + \lambda_{p_i})^{p_i}\ \cdots\ (y_{n_c}^T z + \lambda_{p_i})^{p_i}\ \cdots\ 1]^T,
\end{split}
\end{equation}
where, $p_i\in\{1,\dots,n_p\}$. 
The optimization formulation \eqref{posnegLP} still applies.
Repeating the earlier analysis will show that \eqref{posnegLP} has
at least one point in the solution space for separating hyperplanes (and one for separating hyperellipsoids). Consequently, \eqref{posnegLP} has a non-empty solution space.


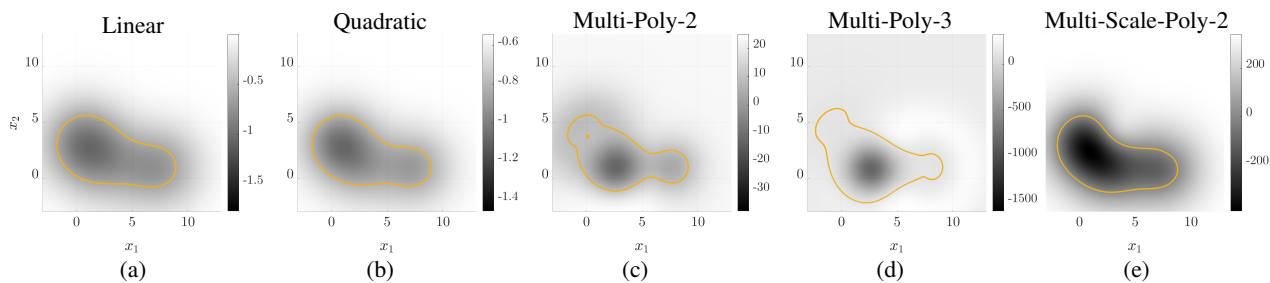
\begin{figure*}[t]
  \centering
  \setlength{\bht}{1.15in}
  \setlength{\xsh}{3pt}
  \setlength{\ylabs}{0pt}
  \vspace*{5pt}
  \begin{tikzpicture}[inner sep=0pt, outer sep=0pt]
    \node (BLin) at (0in,0in)
      {\includegraphics[height=\bht]{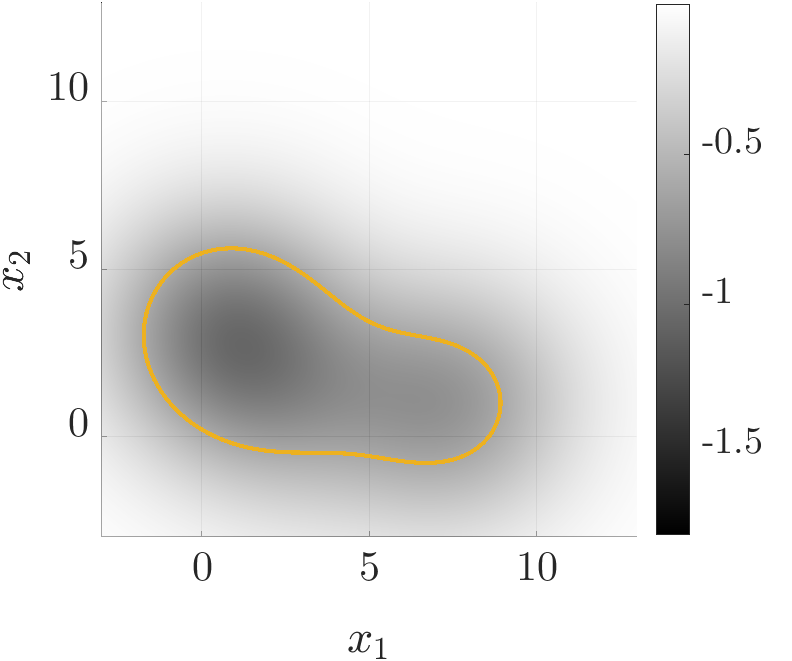}};
    \node[anchor=west] (BQuad) at ([xshift=\xsh]BLin.east)
      {\includegraphics[height=\bht,clip=true,trim=0.675in 0in 0in 0in]{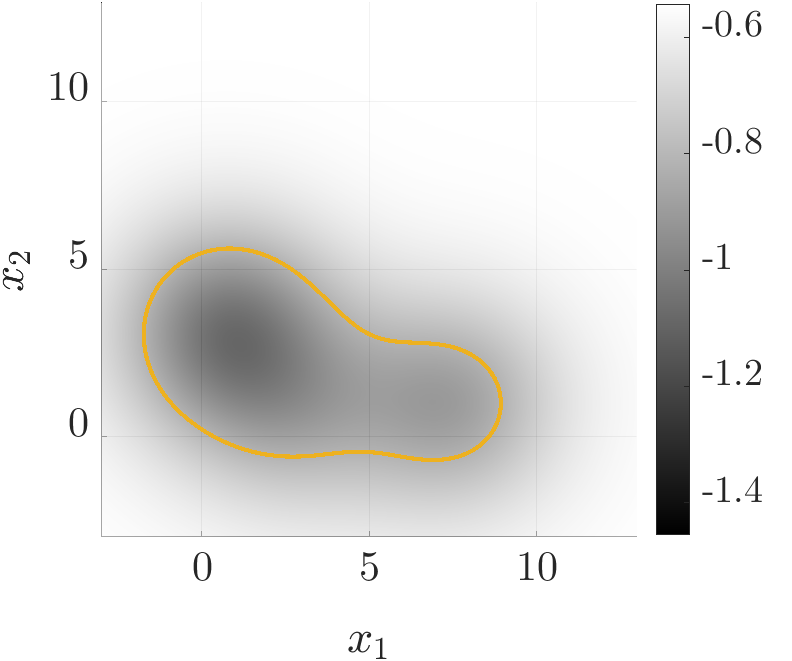}};
    \node[anchor=west] (B2Quad) at ([xshift=\xsh]BQuad.east)
      {\includegraphics[height=\bht,clip=true,trim=0.675in 0in 0in 0in]{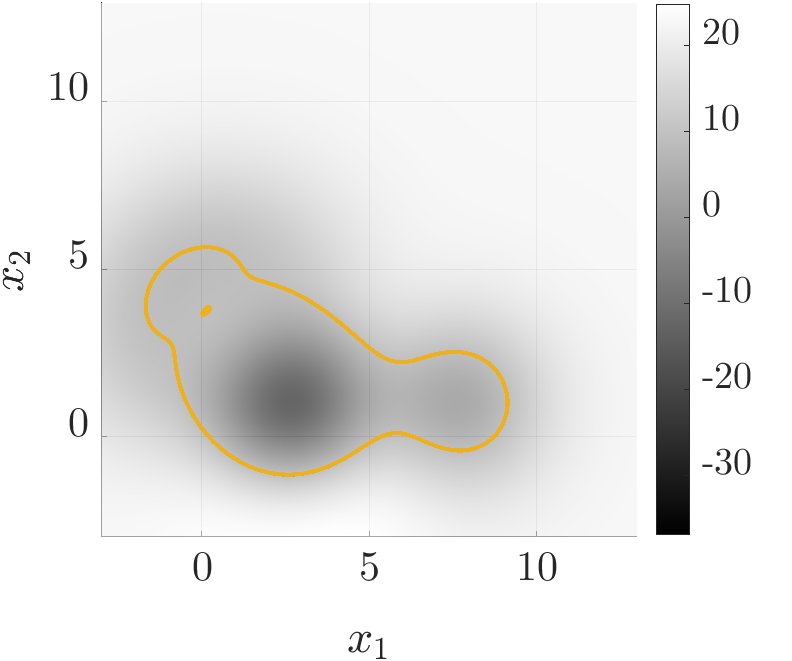}};
    \node[anchor=west] (B2Cube) at ([xshift=\xsh]B2Quad.east)
      {\includegraphics[height=\bht,clip=true,trim=0.675in 0in 0in 0in]{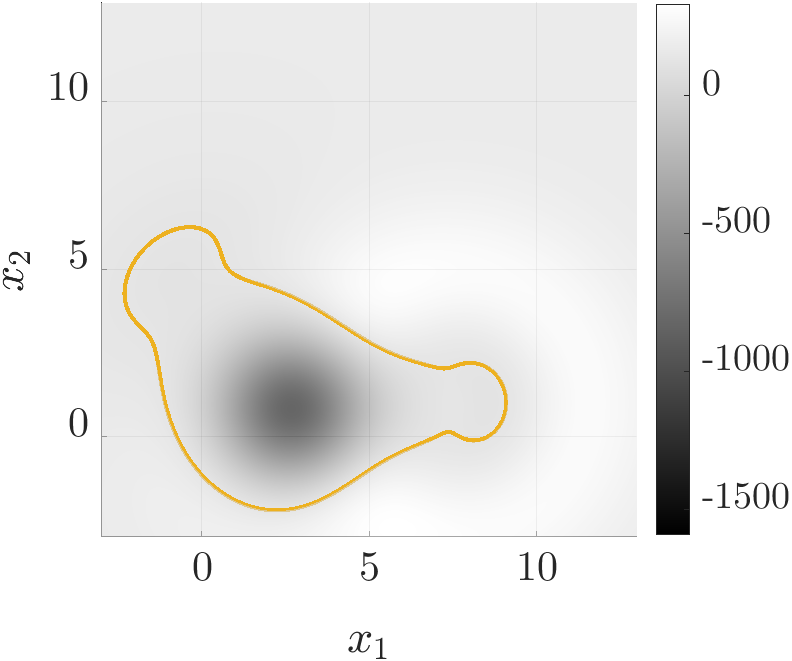}};
    \node[anchor=west] (B2Quad_rbf) at ([xshift=\xsh]B2Cube.east)
      {\includegraphics[height=\bht,clip=true,trim=1.4in 0in 0in 0in]{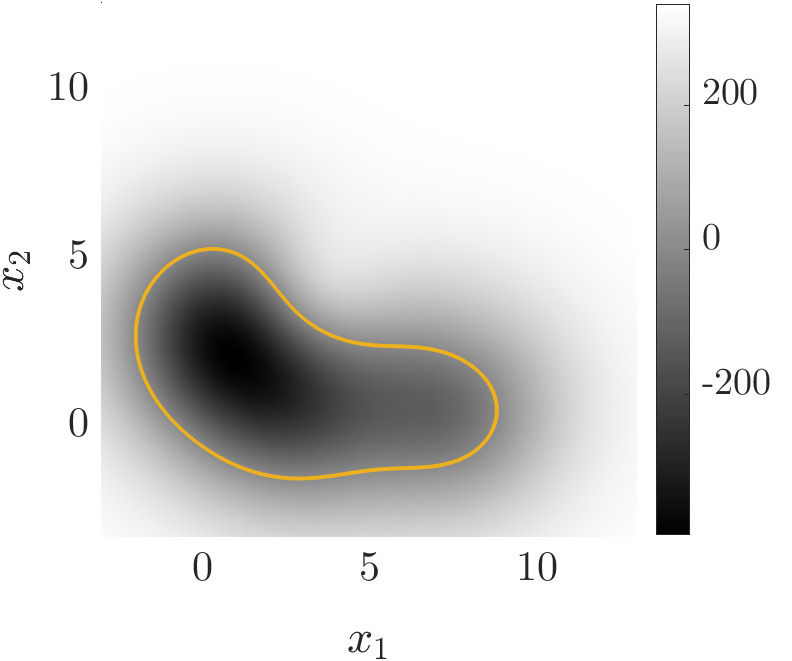}};

    \node[anchor=south,xshift=-3pt,yshift=\ylabs] at (BLin.north)
      {\small Linear};
    \node[anchor=south,xshift=-9pt,yshift=\ylabs] at (BQuad.north)
      {\small Quadratic};
    \node[anchor=south,xshift=-9pt,yshift=\ylabs] at (B2Quad.north)
      {\small Multi-Poly-2};
    \node[anchor=south,xshift=-9pt,yshift=\ylabs] at (B2Cube.north)
      {\small Multi-Poly-3};
    \node[anchor=south,xshift=-9pt,yshift=0] at (B2Quad_rbf.north)
      {\small Multi-Scale-Poly-2};

    \node[anchor=north,xshift=-3pt,yshift=-3pt] at (BLin.south)
      {\small (a)};
    \node[anchor=north,xshift=-9pt,yshift=-3pt] at (BQuad.south)
      {\small (b)};
    \node[anchor=north,xshift=-9pt,yshift=-3pt] at (B2Quad.south)
      {\small (c)};
    \node[anchor=north,xshift=-9pt,yshift=-3pt] at (B2Cube.south)
      {\small (d)};
    \node[anchor=north,xshift=-9pt,yshift=-3pt] at (B2Quad_rbf.south)
      {\small (e)};
  \end{tikzpicture}
  \caption{
  Depiction of {\SKMBF} values as a color map along with its zero level-set (yellow contour) for different network constructions. (a) and (b) are replica of Figure \figref{example}(b,c) (top row); (c) and (d) are replica of Figure \figref{exampleTL}(c); (e) is the same as (c) but with multi-scale Gaussian kernels.
  \figlabel{ZBFplots}
  }
  \vspace*{-1.5em}
\end{figure*}

\section{Construction of the Barrier Function \label{sec:buildMap}}
The previous sections described methods to
synthesize level-set functions for approximating the boundary of the safe set from labeled, sampled data.
The level-set functions have the necessary properties to serve as {\ZBFlong}s after a constant shift and rescaling. 
We call the barrier function designed from concatenated kernel machines a \textit{\SKMBFlong} (\textit{\SKMBF}).
An additional extended class $\mathcal{K}$ outer function 
$\psi \in \mathcal{K}_e$ may be needed to adjust the sensitivity based
on the application, e.g., $h(x) = \psi_0 \circ \bar h(x)$, where $\bar
h(\cdot)$ is the \textit{\SKMBF}. 
This section connects the cutting surface layer with ZBF properties and
concludes with some practical considerations for implementing the
\textit{\SKMBF}.


\subsection{Suitability as Zeroing Barrier Functions}
To serve for safety-critical control applications, there are conditions on the
level-set function to be a valid {\ZBF}.

\paragraph{Monotonicity} 
ZBFs must be monotonic when evaluated across the boundary. 
  The specification of the boundary as a cutting surface constraint means that the domain surface and the constraint surface are transverse to each other. The intersection of the two surfaces occur at the zero level-set. Movement on the (Hilbert space) domain surface away from the boundary necessarily has monotonic behavior, locally. 
  The hard constraint in \eqref{posnegLP} for the unsafe samples ensures local monotonicity, at least from the -1 to +1 level-sets (0 to +2 for the equivalent {\SKMBF}) and for some non-trivial band away from their boundaries. Figure \figref{ZBFplots} depicts {\SKMBF} as a color map for various architectures, all of which exhibit monotonicity across the boundary.

\paragraph{Continuously Differentiable}
  The canonical control barrier function constraint used for safety
  depends on the gradient of the barrier function~\cite{Aa[2019]bf}.
  The function needs to be continuously differentiable
  to avoid undesirable behavior in the constraint.  The Gaussian and
  polynomial layers consist of smooth basis functions, thus their
  composition and linear combination is also smooth, thereby resulting in smooth gradients of the {\ZBF}. The gradient has a closed form solution based on equations \eqref{kernMap} and \eqref{p_vec_bias} which can be used in the CBF constraint~\cite{Aa[2019]bf} to enforce safety.
 
\paragraph{Dead Gradients}
  Based on the CBF constraint discussed in~\cite{Aa[2019]bf}, vanishing gradients for $h(x)$ for $x \in \mathcal{D}$ may lead to loss of control~\cite{mohitCDC2019}, which is undesirable in safety critical applications. Monotonicity in the vicinity of the zero level-set guarantees no vanishing gradients in that region.
  Due to the nature of the first layer and the mapping of $\domain$ to a
  compact surface in $\HspaceN{\centerNum}$, the function has known
  limits, which give local extrema.  One such will be the output of $z =
  0$ (points at infinity in the original input space). These points will
  tend to the same {\ZBF} value (visible as nearly constant regions in
  the level-set plots of Fig. \figref{ZBFplots}). 
  At the other extreme will be the center locations; they locally
  maximize distance from $z = 0$ in $\HspaceN{\centerNum}$, and may
  have similar extremizing properties relative to the transverse cutting
  surface.  These points may be local extrema of the {\ZBF}, in which
  case the gradients will vanish there (visible as local minima in the
  level-set plots of Fig. \figref{ZBFplots}).  
  Since the focus of this {\SKMBF}
  construction is on spatially meaningful or location-based {\ZBF}s
  related to navigation, the existence of the aforementioned extrema on
  sets of measure zero do not have a strong impact as for other use
  cases of {\ZBF}s. Due to the possible topology of \textit{unsafe}
  space, attempts to define one signed distance function for all
  disconnected regions will result in an extrema set of zero measure
  (related to junctions in the Voronoi partition).


\paragraph{Guarantees on the Safe Set} 
The ability to capture the interface between \textit{safe} and
\textit{unsafe} regions is a function of the measurement density.
There is an operating assumption of an $\epsilon$-covering for the {\em
unsafe} regions where the $\epsilon$ value is related to the bandwidth
of the Gaussian kernels, $\sigma$.
If the sensor resolution is too coarse to capture boundary
variation, then the {\em unsafe} set may not lie within the resulting
covering and the asserted guarantees cannot hold. It is impossible to
guarantee safety when the sensors cannot measure or sample the space as
needed (unless it is known how much the local structure can
change for small changes in sample location).

The optimization problem guarantees that $h(\sampSetUnsafe) \subset
\Real^-$, which implies the existence of a point-dependent closed
covering $\coverSetClosed(\sampSetUnsafe)$ of the unsafe points for which
$h(\coverSetClosed(\sampSetUnsafe)) \subset \Real^-$.
Assuming that the sensor can indeed measure and capture the local
structure, and the GRBF layer reflects it, then 
$\setUnsafe \subset \coverSetClosed(\sampSetUnsafe) \subset \subD$ 
and 
$\safeS \subset \safeS^* \subset \subD$. 




\subsection{Practical Method for Barrier Function Synthesis}
The standard kernel machine mapping consists of a single spatial scale,
which places a limit on the boundary that can be captured
for a given $\epsilon$-cover. This limit is related to the Fourier
spectral properties of the functions generated by the kernel machine
\cite{GiJoPo[1995]RegTheoryNN}.  Multi-scale implementations improve this
limitation \cite{Wendland[2004]}, with each Gaussian center of the first layer having a bandwidth parameter chosen from a finite set. 

In effect, each $c_i$ has a bandwidth $\sigma_i \in \set{\beta_1, \,\dots\,,
\beta_{\numKernRBF}}$ where $\numKernRBF$ is the number of unique
bandwidths used.
For multiple Gaussian bandwidths the $\epsilon$-cover concept changes
to be a \textit{Gaussian kernel cover} extending the single-bandwidth
version. The same equation \eqref{epsCovKM} is used but the equivalent
balls $\mathcal{B}_{\epsilon_i}(x)$ in the original input space will
have differing radii $\epsilon_{i}$ based on the kernel function
bandwidth $\beta_i$ for the $i^{\text{th}}$ kernel coordinate mapping,
similar to how $p_i$ varies for a multi-order polynomial layer.
These (radius variable) balls should cover the space of interest.  
The cover $\coverSet$ from \eqref{epsCovKM} satisfies
\begin{equation}
  \subD \subset \coverSet = \union_{i} \mathcal{B}_{\epsilon_i}(c_i)
  \quad \text{for} \ 
      \epsilon_i \in \set{\varepsilon_1, \,\dots\,, \varepsilon_{\numKernRBF}},
\end{equation}
and each $\varepsilon_j$ depends on $\beta_j$.
The same linear programs for hyperplanar, hyperellipsoidal, and
multi-order polynomial cutting surfaces apply, as do the existence of
at least one solution in the feasible space for each LP specification. Earlier work \cite{KiVeGr[2013]RedSet} established a greedy selection
policy for recovering a domain covering center set $\centerSet$ for
optimized function approximation.  A similar policy applies to the case
of multi-scale, first layer synthesis from data, modified with finer bandwidth
values based on residual error minimization of the coarser single
layer basis functions. Likewise, regression or function recovery from
Gaussian kernel machines operate best when capturing the deviation from
some known parametric model \cite{ChEtAl[2020]LearnJump}, leading to a 
semi-parametric first layer (fixed + data-adaptive elements).



\section{Examples and Implementation \label{sec:prob_stat}}

\begin{figure*}[t]
  \centering
  \vspace*{5pt}
  \begin{tikzpicture}[inner sep=0pt, outer sep=0pt]
    \tikzset{putright/.style={anchor=west,xshift=0.02\linewidth}}
    \tikzset{putlabel/.style={anchor=south west,xshift=2pt,yshift=2pt}}
    \node (fig_a) at (0in,0in)
      {\frame{\includegraphics[height=0.8in]{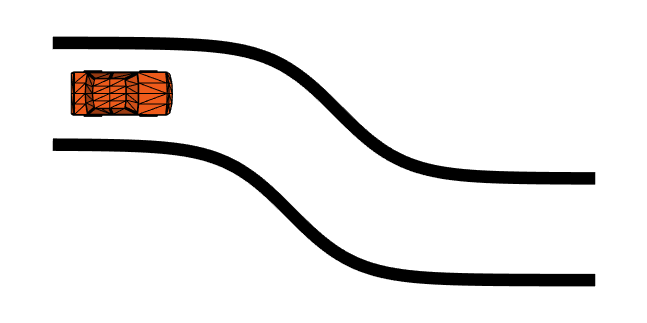}}};
    \node[putright] (fig_b) at (fig_a.east)
      {\frame{\includegraphics[height=0.8in]{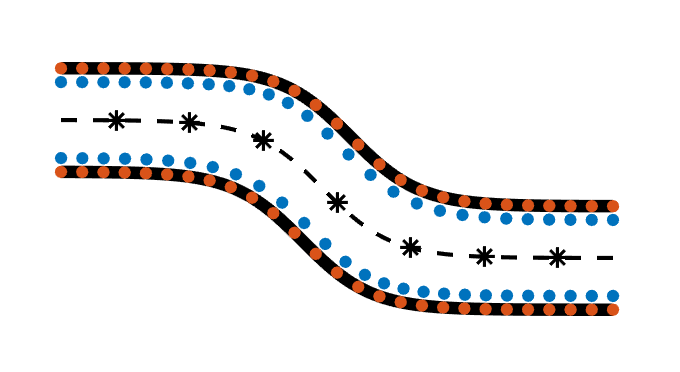}}};

    \node[putlabel] (label_a) at (fig_a.south west){\small (a)};
    \node[putlabel] (label_b) at (fig_b.south west){\small (b)};

    \node[putright] (fig_c) at (fig_b.east)
      {\frame{\includegraphics[height=0.8in]{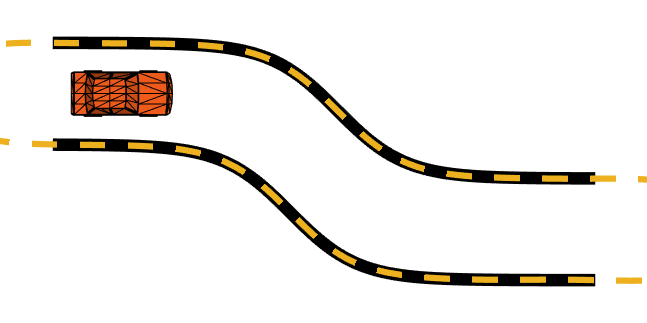}}};
    \node[putright] (fig_d) at (fig_c.east)
      {\frame{\includegraphics[height=0.8in]{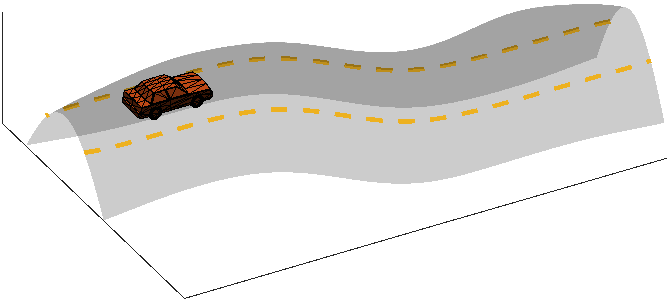}}};

    \node[putlabel] at (fig_c.south west){\small (c)};
    \node[putlabel] at (fig_d.south west){\small (d)};
         
  \end{tikzpicture}
 
 \caption{Curved lane modeling using {\SKMBF}. 
   (a) curved road; 
   (b) road with labels for synthesized safe (blue) and unsafe (red)
     points along with GRBF center locations (black *); 
   (c) synthesized zero level-set of {\SKMBF} (dashed yellow lines); 
   (d) synthesized \textit{\SKMBF} as a 3D surface plus zero level-set.
 \figlabel{laneCase}}

  \vspace*{5pt} 
  \centering
  \begin{tikzpicture}[inner sep=0pt, outer sep=0pt]
    \tikzset{putright/.style={anchor=west,xshift=5pt,draw,dashed,thick}}
    \tikzset{lastone/.style={anchor=west,xshift=0.02\linewidth,draw,thick}}
    \tikzset{putlabel/.style={anchor=south west,xshift=3pt,yshift=3pt}}

    \node[draw,thick,dashed] (fig_a) at (0in,0in)
        {\includegraphics[height=1.05in,angle=-90]{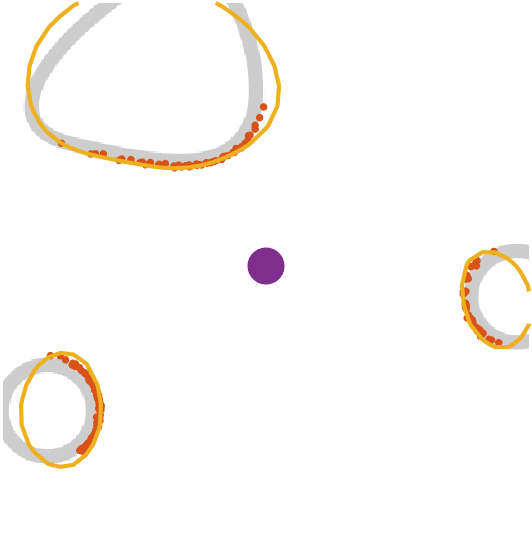}};

    \node[putright] (fig_b) at (fig_a.east)
        {\includegraphics[height=1.05in,angle=-90]{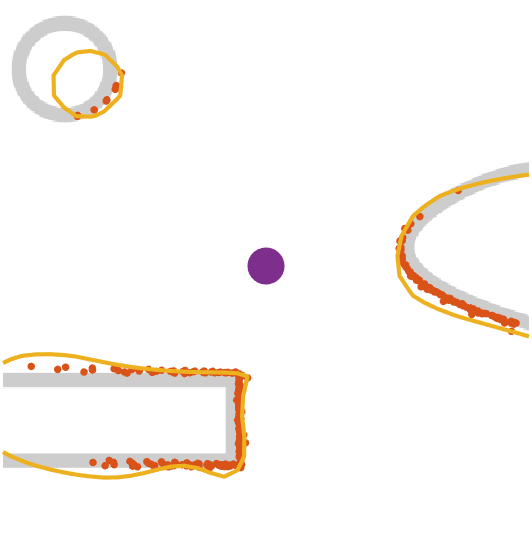}};

    \node[putright] (fig_c) at (fig_b.east)
        {\includegraphics[height=1.05in,angle=-90]{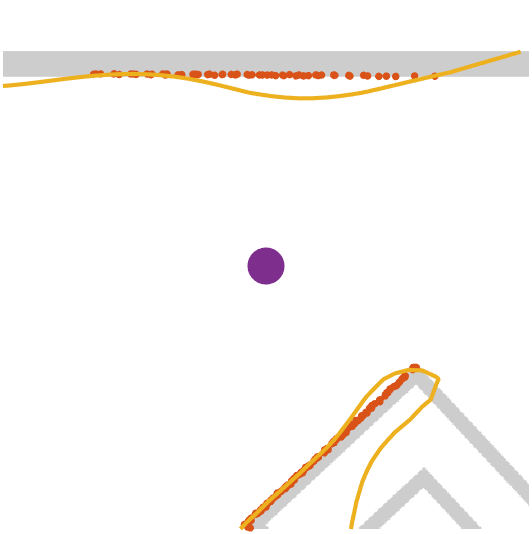}};

    \node[anchor=east,xshift=-5pt] (fig_d) at (fig_a.west)
        {{\includegraphics[width=1.05in,angle=-90]{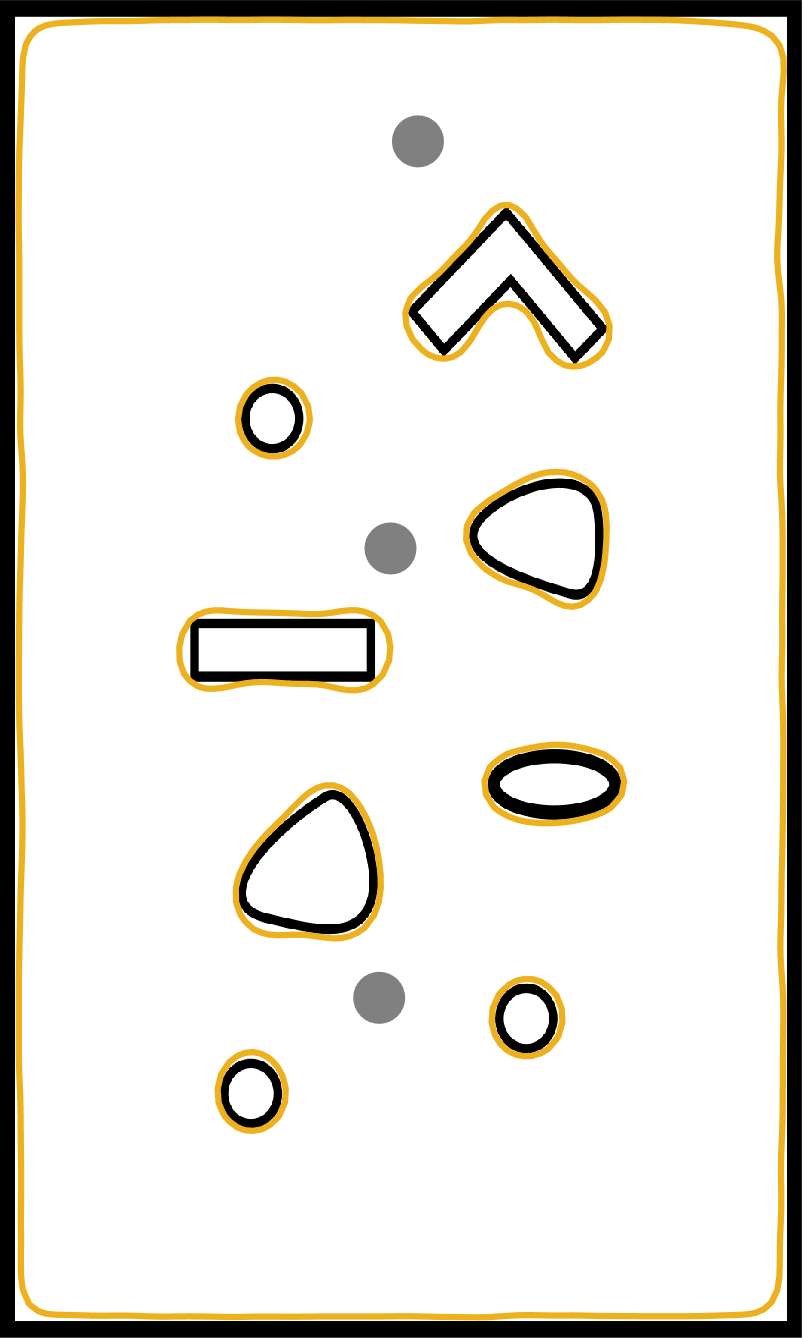}}};
    \node[anchor=east,xshift=-5pt] (fig_e) at (fig_d.west)
        {{\includegraphics[width=1.05in,angle=-90]{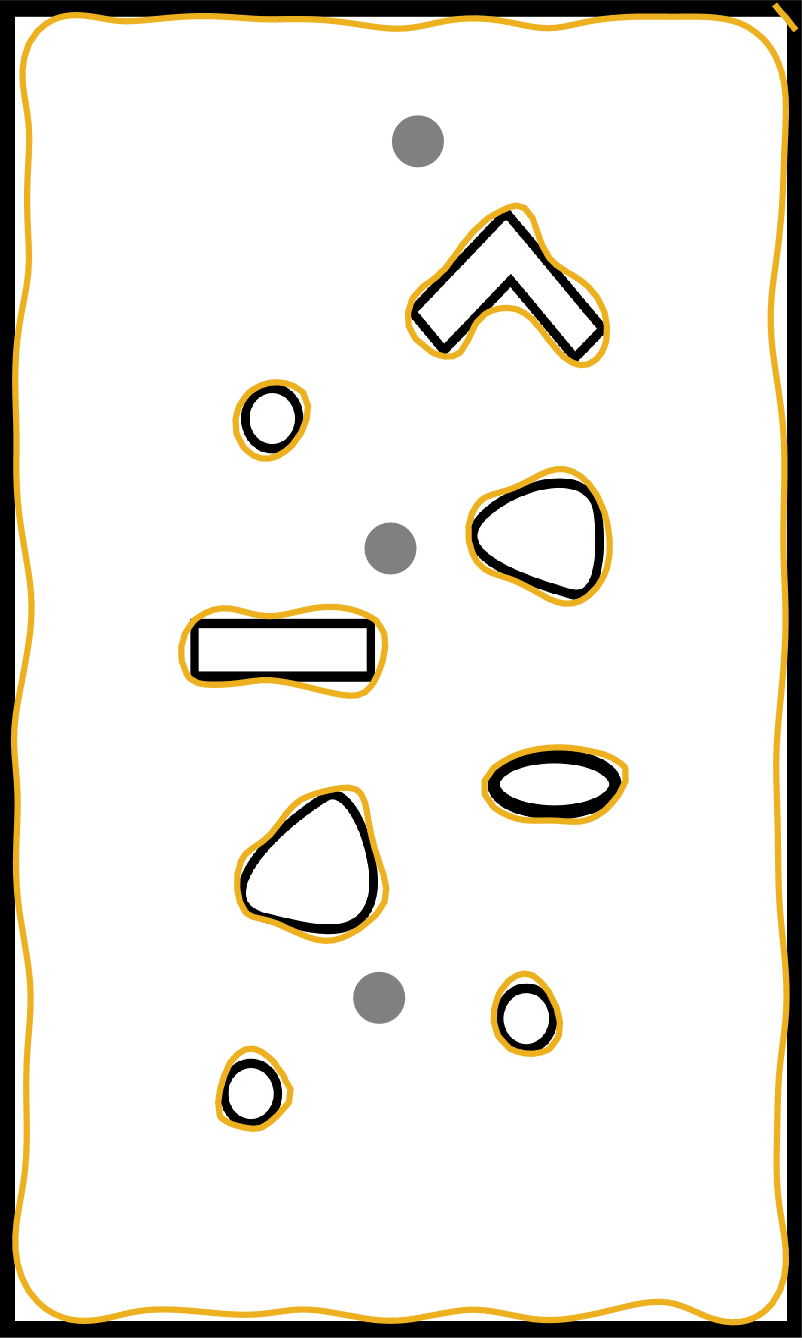}}};

    \node[putlabel] at (fig_a.south west){\small (c)};
    \node[putlabel] at (fig_b.south west){\small (d)};
    \node[putlabel] at (fig_c.south west){\small (e)};
    \node[putlabel] at (fig_d.south west){\small (b)};
    \node[putlabel] at (fig_e.south west){\small (a)};
         
  \end{tikzpicture}
 \caption{Mobile robot planar navigation problem. Synthesis of a global {\SKMBF} using (a) multi-polynomial second layer and (b) Gaussian kernel second layer, to define unsafe region boundaries (yellow contours). Local navigation {\SKMBF} generated at three marked locations (gray dots) in the global maps, are depicted in (c), (d), and (e).  Orange dots are LiDAR points, the purple circle is the robot's location, and the yellow contour is the {\SKMBF} boundary. \figlabel{2DCase}}
 \vspace*{-1.5em}
\end{figure*}


\subsection{Curved Lane Modeling}
This case study shows how the proposed method generates a suitable {\SKMBF} that models curved lanes, as shown in the low-speed driving scenario of Figure \figref{laneCase}(a). Closed-form expressions for curved roads can be challenging to determine, and thus, using control barrier function policy synthesis frameworks for such applications can be difficult. To that end, we show {\SKMBF} synthesis alleviates this issue. The safe and unsafe data points used for network training are shown in Figure \figref{laneCase}(b) and are generated synthetically. In practice, such points would be generated from visual sensing and processing, such as from semantic segmentation of data from a forward-facing camera on the car.

A simple strategy for choosing the GRBFs centers is to uniformly place them along the center line of the lane. Seven GRBFs with a bandwidth of $\sigma=2w_r/\sqrt{\log(2)}$ are used in this case study, as shown in Figure \figref{laneCase}(b) by the black (*); where $w_r$ is the width of the road. The synthesized {\SKMBF} boundary is shown in Figure \figref{laneCase}(c) as the yellow dashed line and as a surface in Figure \figref{laneCase}(d). 
The {\SKMBF} compute time was 10 msec (MATLAB, Ubuntu 20.04, Intel i7-8750H CPU), which supports a 100 Hz update rate.


\subsection{Planar Mobile Robot Navigation in ROS STDR Simulator}
This study involves a differential drive robot equipped with a LiDAR sensor navigating a 2D environment with obstacles, see Figure \figref{2DCase}. 
The robot tracked a predetermined, dynamically feasible, collision-free trajectory. Ideal localization was assumed to be available for the robot. The procedure to generate the dataset of safe and unsafe LiDAR samples was done similar to~\cite{MoVe[2020]bf}.
Global and local {\SKMBF} synthesis was performed.  Global synthesis aggregated all of the sensor data and applied either a multi-polynomial second layer, Figure \figref{2DCase}(a), or a Gaussian second layer via kernel SVM optimization \cite{MoVe[2020]bf}, Figure \figref{2DCase}(b). The latter generates a richer solution space and will be more accurate.
The first layer used a fixed $5\times8$ uniform grid.
Local {\SKMBF} synthesis used a fixed robot-centered $4\times4$ grid for
three locations, shown as gray dots in the global maps. 
The zero level-set of the locally and globally synthesized {\SKMBF s} are shown as yellow curves. Notice that the local {\SKMBF} zero level-set cuts into some sides of obstacles that do not have LiDAR measurements. However, once measurements are taken on those sides the {\SKMBF} will include them.

To further improve the accuracy of the {\SKMBF} in capturing the domain, multi-scale kernels can be added to the first layer in addition to the fixed grid kernels. The centers and bandwidths of the multi-scale kernels are turned according to the adaptive strategy described in section IV.B. 

The global {\SKMBF} compute time was 3.9 sec for the multi-polynomial {\SKMBF} and 4.3 sec for the Gaussian {\SKMBF} when solved using the SMO solver in the LIBSVM package (dual QP formation takes 18 sec).
The local {\SKMBF} compute time averaged 18 msec, which supports a 50 Hz update rate. The training time of the local {\SKMBF} using the proposed LP formulation is significantly shorter than the training time of neural network-based ZBFs. For example, the training time of a Softplus-based deep neural network via gradient descent for local ZBF synthesis was 98 msec on GPU and 407 msec on CPU \cite{KeNi[2022]cbf_sdf}.
\section{Conclusion}

This paper presented a {\ZBFlong} synthesis method based on a two-layer shallow kernel machine network architecture. 
The first layer is constructed from Gaussian kernels whose geometry in the associated Hilbert space was analyzed while considering positive samples (unsafe points) only. The analysis provided theoretical guarantees on the existence of a separating hyperplane, synthesized via a linear program (LP), that splits the Hilbert space into safe and unsafe regions via its level-sets.
A second layer was added to the network that uses polynomial kernels and negative samples (safe points) to reduce misclassifications, while also maintaining the LP formulation. The cutting surface geometry of this second layer also applies to the previous work \cite{MoVe[2020]bf}.
Two case studies demonstrate the efficacy of the {\SKMBF} architecture in generating valid {\ZBFlong}s for motion planning applications. Future work will explore optimized construction of {\SKMBF}s for boundary estimation accuracy in navigation contexts.

\bibliographystyle{ieeetr}
\bibliography{local}

\end{document}